\documentclass{article}
\usepackage[utf8]{inputenc}

\usepackage{amsmath}
\usepackage{amssymb}
\usepackage{amsthm}
\usepackage{xspace}
\usepackage{hyperref}
\usepackage{tikz}
\usepackage{graphicx}
\usepackage{float}
\usepackage{algorithm}
\usepackage[noend]{algpseudocode}
\usepackage{tabularx}
\usepackage[margin=1in]{geometry}
\usepackage{arydshln}
\usepackage{multirow}
\usepackage{wrapfig}
\usepackage{booktabs}
    \usepackage{wrapfig}
\usepackage{caption}
\usepackage{subcaption}

\usepackage{todonotes}
\usepackage{calc,ifthen}
\newcounter{hours}
\newcounter{minutes}
\newcommand{\Printtime}{\setcounter{hours}{\time/60}%
\setcounter{minutes}{\time-\value{hours}*60}%
\thehours:%
\ifthenelse{\value{minutes}<10}{0}{}\theminutes}
\usepackage[nodate]{datetime} 

\usepackage{todonotes}

\usetikzlibrary{positioning,calc,shapes,arrows}
\usetikzlibrary{backgrounds}
\usetikzlibrary{matrix,shadows,arrows} 
\usetikzlibrary{decorations.pathreplacing}
 
\def\etal.{et\penalty50\ al.}
\theoremstyle{plain}
\newtheorem{theorem}{Theorem}[section]
\newtheorem{lemma}[theorem]{Lemma}

\newtheorem{corollary}[theorem]{Corollary}
\newtheorem{fact}[theorem]{Fact}

\theoremstyle{definition}
\newtheorem{definition}{Definition}[section]

\theoremstyle{remark}

\theoremstyle{plain}
\newtheorem*{theorem*}{Theorem}

\renewcommand{\epsilon}{\varepsilon}

\newcommand{\N}{\mathbb{N}}
\newcommand{\cA}{\mathcal{A}}
\newcommand{\cG}{\mathcal{G}}
\newcommand{\ceil}[1]{\ceil#1\rceil} 
\newcommand{\tree}{{\tt TREE}}         
\newcommand{\free}{{\tt FREE}}
\newcommand{\inductive}{{\tt INDUCTIVE}}
\newcommand{\planar}{{\tt PLANAR}}
\newcommand{\treewidth}{{\tt TREEWIDTH}}
\newcommand{\trifree}{{\tt TRIANGLE\textnormal{-}FREE}}

\title{Online Coloring and a New Type of Adversary for Online Graph Problems}
\author{Yaqiao Li\\
Universit\'{e} de Montr\'{e}al\\
\href{mailto:yaqiao.li@umontreal.ca}{yaqiao.li@umontreal.ca}
\and 
Vishnu V. Narayan\\
McGill University\\
\href{mailto:vishnu.narayan@mail.mcgill.ca}{vishnu.narayan@mail.mcgill.ca}
\and 
Denis Pankratov\\
Concordia University\\
\href{mailto:denis.pankratov@concordia.ca}{denis.pankratov@concordia.ca}}

\begin{document}

\maketitle

\begin{abstract}
We introduce a new type of adversary for online graph problems. The new adversary is parameterized by a single integer $\kappa$, which upper bounds the number of connected components that the adversary can use at any time during the presentation of the online graph $G$. We call this adversary ``$\kappa$ components bounded'', or $\kappa$-CB for short. On one hand, this adversary is restricted compared to the classical adversary because of the $\kappa$-CB constraint. On the other hand, we seek competitive ratios parameterized \emph{only} by $\kappa$  with no dependence on the input length $n$, thereby giving the new adversary power to use arbitrarily large inputs.

We study online coloring under the $\kappa$-CB adversary. We obtain finer analysis of the existing algorithms $FirstFit$ and $CBIP$ by computing their competitive ratios on trees and bipartite graphs under the new adversary. Surprisingly, $FirstFit$ outperforms $CBIP$ on trees. When it comes to bipartite graphs $FirstFit$ is no longer competitive under the new adversary, while $CBIP$ uses at most $2\kappa$ colors. We also study several well known classes of graphs, such as $3$-colorable, $C_k$-free, $d$-inductive, planar, and bounded treewidth, with respect to online coloring under the $\kappa$-CB adversary. We demonstrate that the extra adversarial power of unbounded input length outweighs the restriction on the number of connected components leading to non existence of competitive algorithms for these classes.
\end{abstract}

\section{Introduction}
\label{sec:intro}
In online graph problems the input graph is not known in advance, but is rather revealed one item at a time. In this paper we are concerned with the so-called vertex arrival model, where the graph is revealed one vertex at a time. When a new vertex is revealed, an online algorithm learns the identity of the vertex as well as  its neighborhood restricted to the already revealed vertices. Note that the algorithm gets no information about  future vertices. Many graph problems do not admit any non-trivial online algorithms in the adversarial vertex-arrival model. Be that as it may, online graph problems often arise in real life applications in computer networks, public transit networks, electrical grids, and so on. Recently, the interest in online and ``online-like''\footnote{Some examples of ``online-like'' models of computation are dynamic graph algorithms, temporal graph algorithms, streaming graph algorithms, priority graph algorithms, and so on.} graph models and algorithms has been increasing since it is being sparked by the proliferation of online social networks. Thus, it is necessary to introduce various restrictions of the basic adversarial model that allow nontrivial algorithms while capturing interesting real-life scenarios.

One obtains a plethora of restricted adversaries simply by insisting that the adversary generates a graph belonging to a particular family of graphs, such as $\chi$-colorable, planar, $d$-inductive, etc. Another way to relax the classical adversarial model is to consider distributions on graphs and perform average-case analysis. One of the most studied distributions is, of course, the Erd{\"o}s-R{\'e}nyi random graph. While it is mathematically appealing, real life graphs rarely follow this distribution. For example, one of the early empirical observations was that  distributions on degrees of nodes in real social networks are most accurately modeled by power-law distributions~\cite{Newman2010}, whereas the Erd{\"o}s-R{\'e}nyi model induces a binomial distribution on degrees of nodes. Thus, new models of random graphs have been introduced in an attempt to approximate power-law distributions on degrees. Many of these new generative models are inherently offline. A notable exception is the preferential attachment model~\cite{BarabasiA1999}, which perfectly fits within the vertex arrival model. The formal definition is technical, but at a high level this model works as follows. When a new vertex $v$ arrives its neighborhood is generated by connecting $v$ to an already existing vertex $u$ with probability proportional to the current degree of $u$. This model has a natural motivation: consider a person signing up for some social network, where people can ``follow'' each other. This person is signing up not because they want to be left alone (i.e. form a new connected component), but because they already have a list of people in mind who they will follow (i.e., join existing connected component(s), potentially merging some components together). It is more likely that the person is going to follow well known people, e.g. celebrities, who in turn have some of the highest numbers of followers in the network. This is akin to a new node in vertex arrival model likely being connected to existing nodes of high degree. 

The starting point of our work is the observation that when a social network graph is generated via the preferential attachment process, there are very few connected components in the online graph at any point in time. Formalizing this observation in the adversarial setting, we investigate a new type of adversary that is restricted to using at most $\kappa$ connected components at any point in time during the generation of the online input graph. We call such adversary $\kappa$ components bounded, or $\kappa$-CB for short.


In this paper we focus on the online coloring problem under the $\kappa$-CB adversary. Indeed, another motivation for considering the $\kappa$-CB adversary is to extend our understanding of lower bound techniques for online coloring.  Most of the past research uses the following methodology: the adversary creates a collection of disjoint components with some properties, then the adversary merges some of these components by creating a vertex appropriately connected to the components. The aim of this technique is to allow the adversary to observe the coloring of each component chosen by the algorithm, and then choose a ``correct'' coloring of the components that differs from the one chosen by the algorithm. The adversary then connects the components together, forcing the algorithm to use extra colors (since the algorithm's coloring is incorrect inside at least one component). By iterating this process, the adversary tries to force the online algorithm to perform badly. Some variant of this technique has been used, for example, in \cite{gyarfas1988line,vishwanathan1990randomized,bip1,bip2,albers2017tight}.  A notable exception is \cite{halldorsson1992lower}, where this create-and-merge components technique is not directly involved. Usually, this type of construction involves a large number of disjoint components, typically logarithmic in the number of vertices -- see, for example, \cite{gyarfas1988line, bip2}. Our goal is to formally analyze the power of this technique and the extent of dependence of existing lower bounds on this technique. Specifically, we ask, what happens if the adversary in the online coloring problem is $\kappa$-CB? In this work we investigate this question, while allowing the adversary to use an unlimited number of vertices to compensate for a limited number of components.

Our first set of results gives a finer understanding of the $FirstFit$ and $CBIP$ algorithms (for formal definitions see Section~\ref{sec:prelim}), which are well known in the graph coloring community. We show that, perhaps surprisingly, $FirstFit$ outperforms $CBIP$ on trees with respect to the $\kappa$-CB adversary. For general bipartite graphs, we show that $CBIP$ uses at most $2\kappa$ colors against the $\kappa$-CB adversary. This result is particularly interesting in the context of existing lower bounds on the performance of $CBIP$ on bipartite graphs. In a series of works \cite{gyarfas1988line, bip1, bip2} it is shown that any online algorithm must use at least roughly $2\log n$ colors where $n$ is the number of vertices. The construction for this lower bound uses $\log n$ disjoint components. Our result shows that this is necessary. One often measures the performance of an online algorithm by its competitive ratio -- the worst-case ratio between the objective value achieved by an algorithm and the offline optimum. In the case of nontrivial bipartite graphs offline optimum is $2$, thus the difference between the absolute number of colors used by $CBIP$ and its competitive ratio is just a factor of $2$. But this difference has a philosophical significance:  our result shows that the competitive ratio of $CBIP$ on bipartite graphs is simply $\kappa$: the number of components that the adversary is allowed to use.

Our second set of results shows that for several classes of graphs, including $\chi$-colorable graphs,  the $\kappa$-CB adversary equipped with unlimited number of vertices is powerful enough to rule out competitive algorithms even when $\kappa = 1$. These two sets of results provide another contrast between bipartite graphs and other classes of graphs. 

The rest of the paper is organized as follows. In Section~\ref{sec:prelim} we go over some preliminaries. The new adversarial model is introduced in Section~\ref{sec:newmodel}. The $FirstFit$ algorithm is analyzed in Section~\ref{sec:firstfit}, while $CBIP$ is analyzed in Section~\ref{sec:CBIP}. The analysis of various classes of graphs can be found in Section~\ref{sec:classes}. We finish with some discussion and open problems in Section~\ref{sec:conclusion}.

\section{Preliminaries}
\label{sec:prelim}

In online coloring, an adversary creates a simple undirected graph\footnote{We will always use $n$ to refer to $|V|$ and $m$ to refer to $|E|$.} $G=(V,E)$ and a presentation order of vertices\footnote{Notation $[n]$ stands for $\{1, 2, \ldots, n\}$. More generally, notation $[k,n]$ stands for $\{k, k+1, \ldots, n\}$} $\sigma : [n] \rightarrow V$. The graph is then presented online in the \emph{vertex arrival model}: at time $i$ vertex $v=\sigma(i)$ arrives, and we learn all its neighbors among already appearing vertices. An online algorithm must declare how to color the new vertex $c(v)$ prior to the arrival of the next vertex $\sigma(i+1)$. A priori, an online algorithm does not know $V$ or even $n$. Alternatively, we can view the online input as a sequence of induced subgraphs:
\[ G \cap \sigma([1]), G \cap \sigma([2]), \ldots, G \cap \sigma([n]).\]
We call the set of neighbors of $v$ that an online algorithm learns about at the time of arrival of $v$ as the \emph{pre-neighborhood} of $v$, denoted by $N^-(v)$.

A natural greedy algorithm is called $FirstFit$ (see, for example,~\cite{gyarfas1988line,kierstead1998survey,Steffen2014}): when a vertex $v$ arrives, $FirstFit$ colors it with the first color that does not appear in the pre-neighbrhood of $v$. The pseudocode is shown in Algorithm~\ref{algo:ff}.

\begin{algorithm}[!h]
\caption{The $FirstFit$ algorithm for online coloring.}\label{algo:ff}
\begin{algorithmic}
\Procedure{$FirstFit$}{}
\State{$i \gets 1$}
\State{$c \gets \emptyset$}\Comment{map that stores the coloring}
\While{$i \le n$}
\State{A new vertex $v = \sigma(i)$ arrives with its pre-neighborhood $N^-(v)$}
\State{$c(v) \gets \min \left(\mathbb{N} \setminus c(N^-(v))\right)$}
\State{$i \gets i+1$}
\EndWhile
\Return{$c$}
\EndProcedure
\end{algorithmic}
\end{algorithm}

Another famous algorithm due to~\cite{lovasz1989line} for online coloring of \emph{bipartite graphs} is called $CBIP$: when a vertex $v=\sigma(i)$ arrives, $CBIP$ computes an entire connected component $CC$ to which $v$ belongs in the partial graph known so far. Since we assume that the input graph is bipartite, the connected component $CC$ can be partitioned into two sets of vertices $A$ and $B$ such that all edges go between $A$ and $B$ only. Suppose that $v \in A$, then $v$ is colored with the first color that is not present among vertices in $B$. The pseudocode is shown in Algorithm~\ref{algo:cbip}.

\begin{algorithm}[!h]
\caption{The $CBIP$ algorithm for online coloring of bipartite graphs.}\label{algo:cbip}
\begin{algorithmic}
\Procedure{$CBIP$}{}
\State{$i \gets 1$}
\State{$c \gets \emptyset$}\Comment{map that stores the coloring}
\While{$i \le n$}
\State{A new vertex $v = \sigma(i)$ arrives with its pre-neighborhood $N^-(v)$}
\State{$CC \gets$ connected component of $v$ in $G \cap \sigma([i])$}
\State{Partition vertices of $CC$ into $A$ and $B$ such that all edges go between $A$ and $B$ only and $v \in A$}
\State{$c(v) \gets \min \left(\mathbb{N} \setminus c(B)\right)$}
\State{$i \gets i+1$}
\EndWhile
\Return{$c$}
\EndProcedure
\end{algorithmic}
\end{algorithm}

In the adversarial arguments presented in this work we often need to control the chromatic number of constructed instances. These instances can get quite complicated and computing their chromatic number exactly might be rather difficult. The following technique is widely used in the online coloring community, see e.g.,  \cite{halldorsson1992lower}. The adversary is not only going to construct an online instance, but it will also maintain a valid coloring of that instance. Thus, when specifying the adversary we need to define not only how the next input item is generated, but also how it is colored. The key idea is that since the adversary knows and controls how future input items will be generated, it can anticipate its own moves and create a much better coloring than what an online algorithm can achieve without this knowledge. Unless explicitly stated otherwise, by an ``algorithm'' we always mean a \emph{deterministic algorithm}.

\subsection{Bins vs. Colors}    \label{sec:term}
We adopt the terminology introduced in \cite{halldorsson1992lower}: when there is a possibility of ambiguity we say that an online algorithm colors with \emph{bins} and the adversary with \emph{colors} in order to distinguish the two. Let $v$ be a vertex. We use the notation $b(v)$ to denote the bin that is assigned to $v$ by an online algorithm, and $c(v)$ to denote the color that is assigned to $v$ by the adversary. The functions $b$ and $c$ naturally extend to be defined over sets of vertices. Let $A$ be a set of vertices. Define $b(A) = \{b(v) : v\in A\}$ and $c(A) = \{c(v) : v\in A\}$. Sometimes, we say that bin $b$ contains $v$ to mean that $b(v)=b$. 


\subsection{Saturated Bins}  \label{sec:saturate}
We define a notion that is inspired by several previous works \cite{vishwanathan1990randomized,halldorsson1992lower,bip2}. 

\begin{definition}  \label{def:saturated}
Suppose that the adversary is constructing a $\chi$-colorable graph. A bin $b$ is said to be \emph{$p$-saturated} if there are  $p$ vertices  $v_1, \ldots, v_p$  such that 
\begin{equation}   \label{eq:saturate}
b(v_1) = \cdots = b(v_p) = b; 
\quad
|\{c(v_1), c(v_2), \ldots, c(v_p) \} | = p.
\end{equation}
A bin $b$ is said to be \emph{perfectly $p$-saturated} if bin $b$ is $p$-saturated and it contains exactly $p$ vertices. When a bin $b$ is \emph{$\chi$-saturated}, we simply say bin $b$ is \emph{saturated}. 
\end{definition}

The following simple fact demonstrates why this notion might be interesting. 

\begin{fact}\label{fact:saturation} If $t$ bins are all saturated, then every color class contains $t$ vertices in distinct bins. By connecting a new vertex to these $t$ vertices, the algorithm is forced to use a new bin.
\end{fact}

The notion of saturated bins is already implicit in some previous works. For example, the so-called \emph{two-sided colors}\footnote{What is called by ``colors'' in \cite{bip2}, as in many other works, is what we call bins.} in \cite{bip2} are saturated bins when $\chi=2$. The construction in \cite{bip2} forces many two-sided colors, i.e., many saturated bins. The 
proof 
of a lower bound in \cite{halldorsson1992lower}, which we mentioned earlier as an example that does \emph{not} use the common create-and-merge components strategy directly, could be summarized as follows. The adversary has a strategy to force any algorithm to use perfectly $p$-saturated bins for $1 \le p \le \chi /2$. In \cite{vishwanathan1990randomized}, the lower bound construction does not necessarily force $p$-saturated bins, but seeks to create a situation where a weaker form of the Fact~\ref{fact:saturation} is bound to appear.


In Section~\ref{sec:firstfit}, we show that $\kappa$-CB adversary can successively force saturated bins on $FirstFit$ for $\chi$-colorable graphs. This leads to the algorithm being noncompetitive. Explicit examples of saturated bins can be seen in Figure~\ref{fig:thm-FF} in Section~\ref{sec:firstfit}. The construction of forcing saturated bins on $FirstFit$ is generalized in Section~\ref{sec:classes} to work for all algorithms.


\section{A New Type of Adversary}   \label{sec:newmodel}
Let $cc(G)$ denote the number of connected components of graph $G$. 

\begin{definition}  \label{def:beta}
An adversary is said to be \emph{$\kappa$ components bounded}, or $\kappa$-CB, if the input graph $G$ and the presentation order $\sigma$ satisfy
\begin{equation}  \label{eq:condition-kCB}
\forall\ i \in [n] \:\:\: cc(G\cap \sigma([i])) \le \kappa.
\end{equation}

Let $\cA$ denote a deterministic online coloring algorithm. Define $\beta(\cA,\kappa, G)$ to be the maximal number of bins $\cA$ has to use when a $\kappa$-CB adversary constructs the graph $G$. Let $\cG$ denote a class of graphs. Define
\begin{equation}   \label{eq:def-beta-A-kappa-G}
\beta(\cA, \kappa, \cG) = \sup_{G \in \cG} \beta(\cA, \kappa, G),
\end{equation}
and
\begin{equation}   \label{eq:def-beta-kappa-G}
\beta(\kappa, \cG) = \inf_{\cA} \beta(\cA, \kappa, \cG).
\end{equation}
\end{definition}

Let $\chi \in \N$. By identifying $\chi$ with the class of graphs that are $\chi$-colorable, the notation $\beta(\cA,\kappa,\chi)$ and $\beta(\kappa,\chi)$ are defined via \eqref{eq:def-beta-A-kappa-G} and \eqref{eq:def-beta-kappa-G}, respectively.  Let \tree\ denote the class of graphs that are trees. 

Different from traditional online coloring models, a feature of the $\kappa$-CB adversary model is that  the number of vertices of a graph is \emph{not necessarily}  a parameter in the model. In this work, the $\kappa$-CB adversary is allowed to construct graphs with arbitrarily many vertices. We will be interested in understanding what the power and limitations are for an adversary who  can use unlimited number of vertices but is $\kappa$-CB.

\begin{lemma}  \label{lem:infty}
Let $\Omega$ denote the set of all possible graphs and let $\cA$ be an arbitrary algorithm. Then, 
\[\beta(\cA, 1, \Omega) = \infty.\]
\end{lemma}

\begin{proof}
For every $n\in\N$, let $K_n$ denote the complete graph on $n$ vertices. Obviously, an adversary can present $K_n$ in any presentation order while maintaining a single connected component. Thus, $\beta(\cA,1,K_n) \ge n$. As $n$ is arbitrary, the result follows.
\end{proof}

Hence, the $\kappa$-CB adversary model becomes interesting when we consider special classes of graphs. For example, $\beta(FirstFit, \kappa, 2)$ denotes the maximal number of bins that the $\kappa$-CB adversary can force $FirstFit$ (see Algorithm \ref{algo:ff}) to use by constructing a bipartite graph.

\section{$FirstFit$ on $\chi$-colorable Graphs, Triangle-Free Graphs, and Trees}
\label{sec:firstfit}

In this section, we completely characterize the performance of $FirstFit$ on $\chi$-colorable graphs, triangle-free graphs, and trees for $\kappa$-CB adversaries. We begin with the following theorem, which completely determines $\beta(FirstFit, \kappa, \chi)$ for all $\kappa, \chi \in \N$.

\begin{theorem}  \label{thm:FF-bipartite}
Let $\kappa \in \N, \chi \in \N$.
\begin{enumerate}
    \item[(1)] $\beta(FirstFit, \kappa, 1) = 1$ for every $\kappa \ge 1$; 
    \item[(2)] $\beta(FirstFit, 1, 2) = 2$;
    \item[(3)] $\beta(FirstFit,2,2)= \infty$
    \item[(4)] $\beta(FirstFit, \kappa, \chi) = \infty$ for every $\kappa \ge 1$ and $\chi \ge 3$.
\end{enumerate}
\end{theorem}

\begin{proof} \ 

\begin{enumerate}
\item[(1)] The graphs that are $1$-colorable and can be presented by a $\kappa$-CB adversary consist of up to $\kappa$ isolated vertices. Clearly, $FirstFit$ uses a single bin on such graphs.

\item[(2)] A simple induction shows that $FirstFit$ maintains a valid $2$-coloring when a bipartite graph is revealed by a $1$-CB adversary. Base case of a single vertex is trivial. In the inductive step, the newly arriving vertex has edges going to the ``opposite'' side of the bipartition (due to the $1$-CB restriction). By the inductive assumption, those neighbors have been assigned to a single bin consistent with a valid $2$-coloring, $FirstFit$ correctly identifies the other bin for the new vertex.

\item[(3)] Let $n$ be even and consider the graph $G = (V,E)$ with the vertex set $V = \{v_i \mid i \in [n]\}$ and the edge set defined by connecting each vertex $v_{2k-1}$ with $v_{2k'}$ for $k, k' \in [n/2]$ and $k \neq k'$.  Since all edges go in between odd-indexed and even-indexed vertices, the graph is clearly bipartite. The adversary presents vertices in order $v_1, v_2, v_3, \ldots, v_n$. This presentation order satisfies the $2$-CB constraint: $v_1$ is initially in one connected component, when $v_2$ arrives it is isolated and forms the second component, and every future vertex has an edge either to $v_1$ or $v_2$. An example of this graph for $n = 8$ is shown in  Figure~\ref{fig:thm-FF-3}. 

Let $b(v)$ denote the bin that $FirstFit$ assigns $v$ to. We show by induction on $k \in [n/2]$ that $b(v_{2k-1}) = b(v_{2k}) = k$. The base case of $k = 1$ is trivial since when $v_1$ and $v_2$ arrive they are isolated nodes, so they are placed in bin $1$ by $FirstFit$. In the inductive step, $v_{2k-1}$ is connected to $v_{2k'}$ for $k' \in [k-1]$. By induction, $b(v_{2k'}) = k'$, therefore $FirstFit$ assigns $v_{2k-1}$ to a new bin $k$. A similar argument holds for the next arriving vertex $v_{2k}$.

As $n$ can be arbitrary large, the result follows.


\item[(4)] Since $\beta(FirstFit, \kappa, \chi)$ is non-decreasing with respect to both $\kappa$ and $\chi$, it suffices to show that $\beta(FirstFit, 1, 3) = \infty$. Let $n$ be a multiple of $3$ and consider the graph $G=(V,E)$ with the vertex set $V = \{v_i \mid i \in [n]\}$. The adversary presents vertices in order $v_1, \ldots, v_n$. The edge set is defined by the following construction. As usual, we let $c(v)$ denote the color that the adversary maintains for vertex $v$, while $b(v)$ denotes the bin used by $FirstFit$.

The construction consists of two phases: the initial phase and the inductive phase. During the initial phase, the adversary presents a path of $6$ vertices: $v_1, v_2, v_3, v_4, v_5, v_6$. Clearly, when $v_1$ arrives it is an isolated vertex, but each subsequent vertex $v_i$ is revealed with a single edge to $v_{i-1}$. This clearly satisfies the $1$-CB constraint. The adversary assigns colors $c(v_1) = c(v_4) = \text{red}$, $c(v_2) = c(v_5) = \text{green}$, and $c(v_3)=c(v_6)=\text{blue}$. This is a valid $3$-coloring (although $2$ colors are sufficient, the adversary uses more colors in anticipation of the inductive phase). $FirstFit$ assigns bins $b(v_1) = b(v_3)=b(v_5)=1$ and $b(v_2)=b(v_4)=b(v_6)=2$. The result of the initial phase is that $FirstFit$ ends up with $2$ bins that are $3$-saturated.

The inductive phase proceeds in rounds. In round $k \in [3, n/3]$, the vertex $v_{3k-2}$ is revealed and its pre-neighborhood consists of vertices $v_{3k'-1}$ for $k' \in [k-1]$. Then the vertex $v_{3k-1}$ is revealed and its pre-neighborhood consists of vertices $v_{3k'}$ for $k' \in [k-1]$. Lastly, the vertex $v_{3k}$ is revealed and its pre-neighborhood consists of vertices $v_{3k'-2}$ for $k' \in [k-1]$. The adversary assigns $c(v_{3k-2}) = \text{red}, c(v_{3k-1})=\text{green},$ and $c(v_{3k})=\text{blue}$. By a straightforward induction, $FirstFit$ assigns $b(v_{3k-2})=b(v_{3k-1})=b(v_{3k})=k$: prior to round $k$, $FirstFit$ has $k-1$ bins that are $3$-saturated; during round $k$, $FirstFit$ creates a new bin and places all three new vertices into that bin making it $3$-saturated. The initial phase described in the previous paragraph establishes the base case of the induction.

The coloring maintained by the adversary is easily seen to be valid, since the color classes consist of vertices whose indices have the same remainder $\bmod\;3$, while the edges are present only between two vertices whose indices have different  remainders $\bmod\;3$. The $1$-CB constraint is clearly maintained during the inductive phase.

An example of this construction for $n = 12$ is shown in Figure~\ref{fig:thm-FF-4}. 
\end{enumerate}
\end{proof}

\begin{figure}[ht!]
\centering
\begin{subfigure}{.48\textwidth}
  \centering
\begin{tikzpicture}
\draw [red,fill] (-1,0) circle [radius=0.1];
\draw [red,fill] (-1,1) circle [radius=0.1];
\draw [red,fill] (-1,2) circle [radius=0.1];
\draw [red,fill] (-1,3) circle [radius=0.1];

\draw [green,fill] (1,0) circle [radius=0.1];
\draw [green,fill] (1,1) circle [radius=0.1];
\draw [green,fill] (1,2) circle [radius=0.1];
\draw [green,fill] (1,3) circle [radius=0.1];

\draw (-1,0) -- (1,1);
\draw (-1,0) -- (1,2);
\draw (-1,0) -- (1,3);

\draw (-1,1) -- (1,0);
\draw (-1,1) -- (1,2);
\draw (-1,1) -- (1,3);

\draw (-1,2) -- (1,0);
\draw (-1,2) -- (1,1);
\draw (-1,2) -- (1,3);

\draw (-1,3) -- (1,0);
\draw (-1,3) -- (1,1);
\draw (-1,3) -- (1,2);

\node [left] at (-2,0) {$b=1$};
\node [left] at (-2,1) {$b=2$};
\node [left] at (-2,2) {$b=3$};
\node [left] at (-2,3) {$b=4$};

\node [left] at (-1,0) {$v_1$};
\node [right] at (1,0) {$v_2$};
\node [left] at (-1,1) {$v_3$};
\node [right] at (1,1) {$v_4$};
\node [left] at (-1,2) {$v_5$};
\node [right] at (1,2) {$v_6$};
\node [left] at (-1,3) {$v_7$};
\node [right] at (1,3) {$v_8$};

\node at (-1,4) {$c=$ red};
\node at (1,3.95) {$c=$ green};
\end{tikzpicture}
  \caption{An example of the adversarial graph with $n=8$ that is used in the proof of part (3).}
  \label{fig:thm-FF-3}
\end{subfigure}%
\hfill
\begin{subfigure}{.48\textwidth}
  \centering
 \begin{tikzpicture}
\draw [red,fill] (-2,0) circle [radius=0.1];
\draw [red,fill] (-2,1) circle [radius=0.1];
\draw [red,fill] (-2,2) circle [radius=0.1];
\draw [red,fill] (-2,3) circle [radius=0.1];

\draw [green,fill] (0,0) circle [radius=0.1];
\draw [green,fill] (0,1) circle [radius=0.1];
\draw [green,fill] (0,2) circle [radius=0.1];
\draw [green,fill] (0,3) circle [radius=0.1];

\draw [blue,fill] (2,0) circle [radius=0.1];
\draw [blue,fill] (2,1) circle [radius=0.1];
\draw [blue,fill] (2,2) circle [radius=0.1];
\draw [blue,fill] (2,3) circle [radius=0.1];

\node [left] at (-3,0) {$b=1$};
\node [left] at (-3,1) {$b=2$};
\node [left] at (-3,2) {$b=3$};
\node [left] at (-3,3) {$b=4$};

\node [left] at (-2,0) {$v_1$};
\node [left] at (0,1) {$v_2$};
\node [right] at (2,0) {$v_3$};
\node [left] at (-2,1) {$v_4$};
\node [left] at (0,0) {$v_5$};
\node [right] at (2,1) {$v_6$};
\node [left] at (-2,2) {$v_7$};
\node [left] at (0,2) {$v_8$};
\node [right] at (2,2) {$v_9$};
\node [left] at (-2,3) {$v_{10}$};
\node [left] at (0,3) {$v_{11}$};
\node [right] at (2,3) {$v_{12}$};

\node at (-2,4) {$c=$ red};
\node at (0,3.95) {$c=$ green};
\node at (2,4) {$c=$ blue};

\draw (0,1) -- (-2,0);
\draw (2,0) -- (0,1);
\draw (-2,1) -- (2,0);
\draw (0,0) -- (-2,1);
\draw (2,1) -- (0,0);

\draw (-2,2) -- (0,1);       
\draw (-2,2) -- (0,0);

\draw (0,2) -- (2,1);    
\draw (0,2) -- (2,0);

\draw (2,2) -- (-2,1);    
\draw (2,2) to [out=240,in=20] (-2,0);

\draw (-2,3) -- (0,2);   
\draw (-2,3) -- (0,1);
\draw (-2,3) -- (0,0);

\draw (0,3) -- (2,2);   
\draw (0,3) -- (2,1);
\draw (0,3) -- (2,0);

\draw (2,3) -- (-2,2);   
\draw (2,3)  to [out=240,in=20] (-2,1);
\draw (2,3) -- (-2,0);
\end{tikzpicture}
  \caption{An example of the adversarial graph with $n=12$ that is used in the proof of part (4).}
  \label{fig:thm-FF-4}
\end{subfigure}
\caption{Examples of adversarial inputs used in the proof of Theorem~\ref{thm:FF-bipartite}. Columns indicate the coloring maintained by the adversary, while rows indicate the bins used by $FirstFit$.}
\label{fig:thm-FF}
\end{figure}
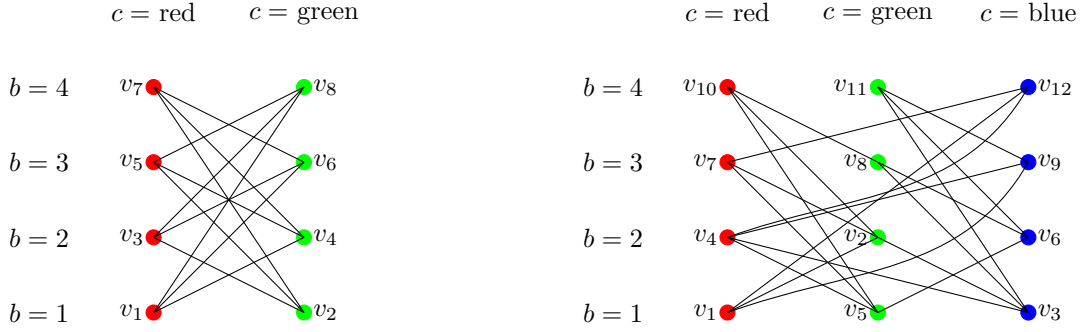

Observe that the construction in part (4) of Theorem \ref{thm:FF-bipartite} results in a triangle-free graph. Let $\trifree$ denote the class of triangle-free graphs. Thus, we immediately obtain the following.
\begin{corollary}  \label{cor:FF}
$\beta(FirstFit, 1, \trifree) = \infty$. 
\end{corollary}

Part (4) of Theorem \ref{thm:FF-bipartite} and Corollary \ref{cor:FF} will be generalized in Section \ref{sec:classes}.  

We conclude this section by giving a complete analysis of $FirstFit$ on trees with respect to $\kappa$-CB adversary.

\begin{theorem}   \label{thm:FF-tree}
$\beta(FirstFit, \kappa, \tree) = \kappa+1$.
\end{theorem}

\begin{proof}
The lower bound $\beta(FirstFit, \kappa, \tree) \ge \kappa+1$ is witnessed by the so-called forest construction due to Bean~\cite{bean1976effective} (also independently discovered in~\cite{gyarfas1988line}). We claim that a $\kappa$-CB adversary can construct a forest consisting of $\kappa$ trees $T_1, T_2, \ldots, T_{\kappa}$ with the property that for each $i$ the  $FirstFit$ algorithm uses color $i$ on some vertex $v_i$ belonging to the tree $T_i$. We first prove this claim and later see how it implies the lower bound.

The construction is recursive and so we prove the above statement by induction on $\kappa$. Base case is trivial: when $\kappa=1$ the adversary can give a single isolated vertex. Assume that the statement is true for $\kappa$ and we wish to establish it for $\kappa+1$. The adversary begins by invoking induction and creating $T_1, T_2, \ldots, T_{\kappa}$ such that for $i \in [\kappa]$ there is $v_i \in T_i$ such that $b(v_i)=i$ is assigned by $FirstFit$. Then, the adversary creates a new vertex $u$ connected to all the $v_i$. This process, merges all existing trees into a single tree, which we call $T'_{\kappa+1}$. Moreover, this forces $FirstFit$ to assign $b(u) = \kappa+1$. This tree is set aside, and to satisfy the claim for $\kappa+1$, the adversary invokes the induction again to create another set of trees $T'_1, \ldots, T'_\kappa$ with $v'_i \in T'_i$ such that $b(v'_i) = i$. Note that creating $T'_i$ requires at most $\kappa$ components, so the adversary is $(\kappa+1)$-CB (remember that we  have an additional component $T'_{\kappa+1}$ set aside during the second invocation of induction). Moreover, note that $b(u) = \kappa+1$, so the trees $T'_1, \ldots, T'_{\kappa+1}$ satisfy the claim.

This claim implies the lower bound since the $\kappa$-CB adversary can present $T_1, \ldots, T_\kappa$ with $b(v_i) = i$ for some $v_i \in T_i$. In the last step, the adversary presents $u$ connected to each $v_i$. This process does not increase the number of components and forces $b(u) = \kappa+1$. An example of this construction is shown in Figure~\ref{fig:thm-FF-tree}.

Next, we show the upper bound $\beta(FirstFit, \kappa, \tree) \le \kappa+1$. Suppose that $FirstFit$ uses $m+1$ bins for some $m$.
Let $v$ be the first vertex which is placed into bin $m+1$ by $FirstFit$. By the definition of $FirstFit$, there are $m$ vertices in the pre-neighborhood $N^{-}(v)$ that have been previously assigned to bins $1,2,\ldots, m$. 
Since the adversary is constructing a tree, there can be no cycle. Hence, these $m$ vertices must be in distinct components, i.e., there are at least $m$ distinct components. 
\end{proof}

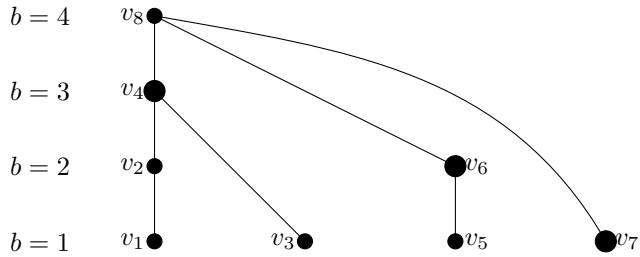
\begin{figure}[ht!]
\centering
\begin{tikzpicture}
\draw [fill] (-3,0) circle [radius=0.1];
\draw [fill] (-1,0) circle [radius=0.1];
\draw [fill] (1,0) circle [radius=0.1];
\draw [fill] (3,0) circle [radius=0.14];
\draw [fill] (-3,1) circle [radius=0.1];
\draw [fill] (1,1) circle [radius=0.14];
\draw [fill] (-3,2) circle [radius=0.14];
\draw [fill] (-3,3) circle [radius=0.1];

\draw (-3,0) -- (-3,1) -- (-3,2) -- (-3,3);
\draw (-1,0) -- (-3,2);
\draw (1,0) -- (1,1) -- (-3,3);
\draw (3,0) to [out=120,in=350] (-3,3);

\node [left] at (-4,0) {$b=1$};
\node [left] at (-4,1) {$b=2$};
\node [left] at (-4,2) {$b=3$};
\node [left] at (-4,3) {$b=4$};

\node [left] at (-3,0) {$v_1$};
\node [left] at (-3,1) {$v_2$};
\node [left] at (-1,0) {$v_3$};
\node [left] at (-3,2) {$v_4$};
\node [right] at (1,0) {$v_5$};
\node [right] at (1,1) {$v_6$};
\node [right] at (3,0) {$v_7$};
\node [left] at (-3,3) {$v_8$};
\end{tikzpicture}
\caption{An example of the forest construction used in the proof of Theorem \ref{thm:FF-tree}. The adversary presents the vertices in order: $v_1, v_2,v_3, \ldots, v_8$. The $FirstFit$  uses $4$ bins while  the adversary uses $3$ connected components during this construction.}
\label{fig:thm-FF-tree}
\end{figure}

\section{$CBIP$ on Bipartite Graphs}   
\label{sec:CBIP}


This section contains the most technical result of this paper. We establish the tight bound of $2\kappa$ on the number of bins used by $CBIP$ with respect to a $\kappa$-CB adversary on bipartite graphs and trees. This provides a finer understanding of the performance of $CBIP$ and is particularly interesting in light of previous lower bounds. Recall that Gutowski et al.~\cite{bip2} proved that any online algorithm has to use at least $2\log n-10$ bins for coloring bipartite graphs with $n$ vertices, which matches the upper bound on $CBIP$ from \cite{lovasz1989line} up to the additive constant $-10$. The construction in \cite{bip2} applied to $CBIP$ (or even $FirstFit$) uses $\log n$ disjoint connected components to force $2\log n$ bins. In particular, the main result of this section, which we state next, demonstrates that this is a necessary feature of their construction.

\begin{theorem}   \label{thm:CBIP-tree}
$\beta(CBIP, \kappa, \tree) = \beta(CBIP, \kappa, 2) = 2\kappa$.
\end{theorem}
\begin{proof}
Since $\beta(CBIP, \kappa, 2) \ge \beta(CBIP, \kappa, \tree)$, the lower bound follows from Lemma~\ref{lem:CBIP-lowerbound} and the upper bound follows from Lemma~\ref{lem:CBIP-upperbound}.
\end{proof}

Observe that Theorem \ref{thm:CBIP-tree} implies that, in the class of bipartite graphs, worst case input already appears in \tree\ for $CBIP$. 

We begin by establishing the lower bound used in the above theorem.

\begin{lemma}
\label{lem:CBIP-lowerbound}
$\beta(CBIP, \kappa, \tree) \ge 2\kappa.$
\end{lemma}

\begin{proof}
In this proof, the notation $r(T)$ is used to denote the root of a rooted tree $T$. The statement of the theorem is witnessed by the following recursive adversarial construction: 

\textit{Base cases:} $T_1$ is a rooted tree consisting of a single vertex. $T_2$ is a rooted tree consisting of one edge, where $r(T_2)$ is defined as the vertex that is assigned to bin $2$ by $CBIP$. 

\textit{Recursive step:} let $i \ge 3$. To construct $T_i$ the adversary does the following:
\begin{itemize}
    \item[(1)] it constructs $T_{i-1}$;
    \item[(2)] it constructs $T_{i-2}$;
    \item[(3)] it presents a new vertex $v$ connected via an edge to $r(T_{i-1})$ and via another edge to $r(T_{i-2})$.
\end{itemize}
The vertex from step (3) becomes the root $r(T_i)$ of the newly formed tree $T_i$.

As usual, let $b(v)$ denote the bin to which $v$ is assigned by $CBIP$. For each $i$, let $E_i$ denote the set of nodes that are at even  distance from $r(T_i)$ in $T_i$. Similarly, let $O_i$ denote the set of nodes at odd distance from $r(T_i)$ in $T_i$. We claim that for the above construction it holds that 
\begin{itemize}
    \item[(i)]$b(E_i) = [i]\setminus\{i-1\}$;
    \item[(ii)] $b(O_i) = [i-1]$;
    \item[(iii)] the construction of $T_i$  satisfies the $\lfloor i/2 \rfloor$-CB constraint.
\end{itemize}

We prove the above claim by strong induction on $i$. The statements are immediate for the base cases of $T_1$ and $T_2$. Assume that the statement holds for all $j \le i-1$ for some $i \ge 3$. Next, consider $T_i$. Examining the construction we obtain:
\begin{itemize}
    \item $E_i = O_{i-2} \cup O_{i-1} \cup \{v\}$,
    \item $O_i = E_{i-2} \cup E_{i-1}$.
\end{itemize}
Using the inductive assumption, we have $b(E_{i-2}) = [i-2]\setminus\{i-3\}$ and $b(E_{i-1}) = [i-1]\setminus\{i-2\}$. Therefore, we have $b(O_i)=b(E_{i-2} \cup E_{i-1}) = [i-1]$ establishing part (ii) of the claim.

Using the inductive assumption again, we have $b(O_{i-2}) = [i-3]$ and $b(O_{i-1}) = [i-2]$. Therefore, $b(E_i \setminus\{v\}) = b(O_{i-2} \cup O_{i-1}) = [i-2]$.

$CBIP$ uses the bipartition induced by $E_{i}$ and $O_{i}$ to decide how to color $v$. Specifically, $CBIP$ uses the first available bin other than those in $O_i$. Therefore, since $b(O_i) = [i-1]$ we have $b(v) = i$. Combining it with the result from the previous paragraph, we get $B(E_i) = [i] \setminus\{i-1\}$ establishing part (i) of the claim.

As for the number of connected components used in the presentation of $T_i$, observe that step (1) uses $\lfloor (i-1)/2 \rfloor$ components by induction. After step (1), $T_{i-1}$ is put aside as a separate component. Therefore, step (2) uses $1 + \lfloor (i-2)/2 \rfloor$ connected components (we invoked the inductive assumption one more time here). Step (3) of the construction does not require any additional components. Therefore, the total number of components is bounded by $\max\left(\lfloor (i-1)/2 \rfloor, 1 + \lfloor (i-2)/2 \rfloor \right).$ It is easy to see that this expression is exactly $\lfloor i/2 \rfloor$ by considering the cases of odd and even $i$ separately. This establishes part (iii) of the claim.

Lastly, note that the claim implies that $CBIP$ uses $2\kappa$ bins on $T_{2\kappa}$ and that the presentation of $T_{2\kappa}$ satisfies the $\kappa$-CB constraint.
\end{proof}

We finish this section with a matching upper bound for the class of bipartite graphs.

\begin{lemma}   \label{lem:CBIP-upperbound}
$\beta(CBIP,\kappa,2) \le 2\kappa$.
\end{lemma}
\begin{proof}
Consider a certain point in execution of $CBIP$ on the input graph. As usual, let $b(v)$ denote the bin to which $v$ is assigned by $CBIP$. We say that a connected component $CC$ is of Type $1[\ell]$ if $CC$ can be partitioned in two blocks $A$ and $B$ such that
\begin{itemize}
    \item all edges go between $A$ and $B$
    \item $b(A) = [\ell-2]$
    \item $b(B) = [\ell-1]$
\end{itemize}

Similarly, we say that a connected component $CC$ is of Type $2[\ell]$ if $CC$ can be partitioned in two blocks $A$ and $B$ such that
\begin{itemize}
    \item all edges go between $A$ and $B$
    \item $b(A) = [\ell-2] \cup \{\ell\}$
    \item $b(B) = [\ell-1]$
\end{itemize}

Figure~\ref{fig:cbip-types} shows an example construction with two connected components of Type $2[4]$ and Type $1[5]$.

\begin{figure}[ht!]
    \centering
    \includegraphics{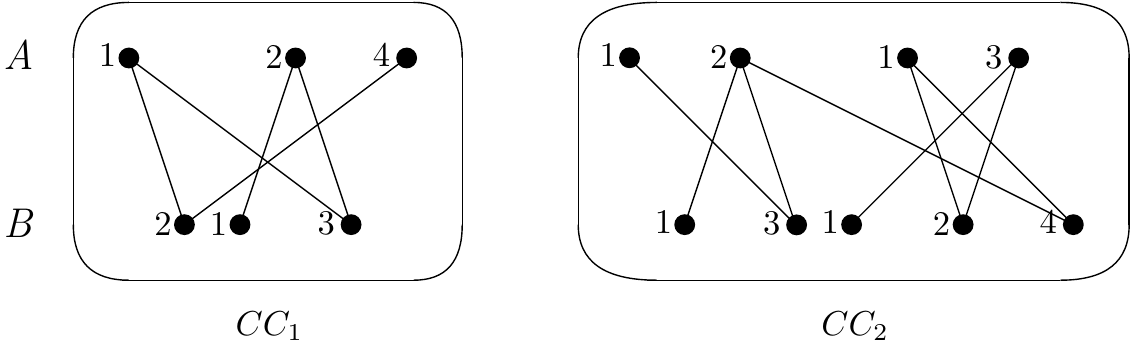}
    \caption{A snapshot of an execution of $CBIP$ on an input instance, with vertices presented in left-to-right order. At this step in the presentation, the graph contains two connected components $CC_1$ and $CC_2$ of Type $2[4]$ and Type $1[5]$ respectively. Labels indicate bins used by $CBIP$.}
    \label{fig:cbip-types}
\end{figure}

The high level idea is that the $\kappa$-CB adversary can only force components that are either of Type $1[\ell]$ or of Type $2[\ell]$ for some $\ell \le 2\kappa$. Before we prove it formally, we observe that when a new vertex $v$ arrives, it can either (1) be an isolated vertex (taken to be of Type $1[2]$), (2) be added to an existing component, or (3) be used to merge two or more existing components. Formally, we say that components $CC_1, CC_2, \ldots, CC_j$ get \emph{merged} at time $t$ if vertex $v = \sigma(t)$ satisfies $N^-(v) \cap CC_i \neq \emptyset$ for $i \in [j]$, and $CC_i$ were distinct connected components at time $t-1$.

We record what happens to types of components after each of the above operations (1), (2) and (3). During operation (1), a new vertex of Type $1[\ell]$ for $\ell=2$ is added. Clearly, $\ell \leq 2\kappa$ for any $\kappa \geq 1$. Next, we consider operation (2), i.e., when a new vertex $v$ gets added to a component $CC$. We assume that the two blocks of vertices of $CC$ are $A$ and $B$ and that they satisfy the conditions of Type $1$ or $2$. The resulting component is called $CC'$. The changes to types after vertex $v$ is presented are recorded in Table~\ref{table:case1}.

\begin{table}[h!]
\caption{Type changes for case (1), i.e., when $v$ is added to an existing component.}  \label{table:case1}
\begin{center}
\begin{tabular}{ |l|l|l| } 
\hline
$v$ is connected to& Type of $CC$ & Type of $CC'$ \\ 
\hline
 \hline
 $A$ & $1[\ell]$ & $1[\ell]$ \\ 
 $A$ & $2[\ell]$ & $2[\ell]$ \\ 
 $B$ & $1[\ell]$ & $2[\ell]$ \\ 
 $B$ & $2[\ell]$ & $2[\ell]$ \\ 
 \hline
\end{tabular}
\end{center}
\end{table}

Finally, we consider what happens when a vertex $v$ is used to merge two or more components. We distinguish four types of components, the numbers of which are denoted by $k_1, k_2, k_3$, and $k_4$, respectively:
\begin{enumerate}
    \item $CC_i^A$ of Type $1[\ell_i^A]$ for $i \in [k_1]$. Vertex $v$ has a neighbor on the $A$-side of such components.
    \item $CC_i^B$ of Type $1[\ell_i^B]$ for $i \in [k_2]$. Vertex $v$ has a neighbor on the $B$-side of such components.
    \item ${CC'}_i^A$ of Type $2[{\ell'}_i^A]$ for $i \in [k_3]$. Vertex $v$ has a neighbor on the $A$-side of such components.
    \item ${CC'}_i^B$ of Type $2[{\ell'}_i^B]$ for $i \in [k_4]$. Vertex $v$ has a neighbor on the $B$-side of such components.
\end{enumerate}


Let $m = \max\{\ell_{i_1}^A, \ell_{i_2}^B, {\ell'}_{i_3}^A, {\ell'}_{i_4}^B: i_1 \in [k_1], i_2 \in [k_2], i_3 \in [k_3], i_4 \in [k_4]\}$. We call $m$ the \emph{type parameter} of the partially constructed input graph.

We say that a block $A$ or $B$ of a particular component being merged is on the opposite side of $v$ if $v$ has a neighbor among the vertices of the block. Otherwise, we say that the block is on the same side as $v$. For example, block $A$ of $CC_i^A$ component is on the opposite side of $v$, whereas block $B$ of the same component is on the same side as $v$. Let $S_{-v}$ denote the set of bins already used for the vertices of blocks on the opposite side of $v$, and let $S_v$ denote the set of bins already used for the vertices of blocks on the same side as $v$. By the definitions of Type 1 and 2 components as well as $m$, it is easy to see that each of $S_v, S_{-v}$ can be only one of the following four options: $[m-2], [m-1], [m-2]\cup \{m\}, [m]$. This reduces the problem of computing the type of the merged component to analyzing 16 cases. For example, if $S_{-v}=[m-2]$ and $S_v = [m-2]$ then vertex $v$ will be assigned bin $m-1$ and the merged component will be of Type $1[m]$ since it will have one side with bins $[m-2]$ and the opposite side with bins $[m-1]$. We denote the merged component by $CC'$ and Table~\ref{table:case2} summarizes all of the 16 cases.

\begin{table}[h!]
\caption{Type changes for case (2), i.e., when $v$ is used to merge some existing components.}   \label{table:case2}
\begin{center}
\begin{tabular}{ |l|l|l|l|l| } 
\hline
$S_{-v}$ & $S_v$ & Possible? & Bin of $v$ & Type of $CC'$ \\ 
\hline
 \hline
 $[m-2]$ & $[m-2]$ & yes & $m-1$ & $1[m]$ \\
 $[m-2]$ & $[m-1]$ & yes &$m-1$ & $1[m]$ \\
 $[m-2]$ & $[m-2]\cup\{m\}$ & no & NA & NA \\
 $[m-2]$ & $[m]$ & no & NA & NA \\
 
 $[m-1]$ & $[m-2]$ & yes & $m$ & $2[m]$ \\
 $[m-1]$ & $[m-1]$ & yes & $m$ & $1[m+1]$ \\
 $[m-1]$ & $[m-2]\cup\{m\}$ & yes & $m$ & $2[m]$ \\
 $[m-1]$ & $[m]$ & yes & $m$ & $1[m+1]$\\
 
 $[m-2]\cup\{m\}$ & $[m-2]$ & no & NA & NA \\
 $[m-2]\cup\{m\}$ & $[m-1]$ & yes & $m-1$ & $2[m]$ \\
 $[m-2]\cup\{m\}$ & $[m-2]\cup\{m\}$ &no & NA & NA \\
 $[m-2]\cup\{m\}$ & $[m]$ & no & NA & NA\\
 
 $[m]$ & $[m-2]$ & no & NA & NA \\
 $[m]$ & $[m-1]$ &yes & $m+1$ & $2[m+1]$ \\
 $[m]$ & $[m-2]\cup\{m\}$ & no & NA & NA \\
 $[m]$ & $[m]$ & yes & $m+1$ & $1[m+2]$ \\
 \hline
\end{tabular}
\end{center}
\end{table}
The reason that certain combinations in  Table~\ref{table:case2} are impossible is that if one side is colored with bins $[m]$ then the opposite side must use bin $m-1$ (because of how $CBIP$ works).  

Observe that from Table~\ref{table:case2}, the type parameter of $CC'$ can either stay the same, increase by additive $1$, or increase by additive $2$.  Furthermore, it can be directly verified  from the table that an increase is possible only if there are at least two components having type parameters not less than $m-1$.  We refer to this property as the \emph{continuity of the type parameter}. 

Assume that the input graph $G$ and the presentation order $\sigma$ satisfy the $\kappa$-CB condition, then the above observations imply the following statements:
\begin{enumerate}
    \item[(i)] $\forall\ i$ we have $G \cap \sigma([i])$ consists of Type $1/2[\ell]$ components for $\ell \le 2 \kappa$;
    \item[(ii)] $\forall\ i$ there can be at most one component that is of one of the following four types: Type $1[2\kappa-1]$, Type $1[2\kappa]$, Type $2[2\kappa-1]$, Type $2[2\kappa]$.
\end{enumerate}
Note that (i) immediately implies the statement of this lemma.



These statements can be proved by induction on $\kappa$. The base case $\kappa=1$ is easy to verify. Indeed, if $i=1$, i.e., there is just a single vertex, then it is of Type 1[2]. Consider $i\ge 2$. Observe that for $\kappa=1$ the algorithms $CBIP$ and $FirstFit$ have identical behavior. Therefore, (2) in Theorem \ref{thm:FF-bipartite} shows (i) is true, and the single component is of Type 2[2]. Hence, (ii) is true.

We proceed to the induction step. Assume (i) and (ii) are true for $\kappa$, we consider the case $\kappa +1$. First, we show (ii). Let $CC_1$ be the first component that is of one of the following types: Type $1[2\kappa+1]$, Type $1[2\kappa+2]$, Type $2[2\kappa+1]$, Type $2[2\kappa+2]$. Since the adversary is $(\kappa+1)$-CB, the existence of $CC_1$ implies that the adversary becomes  $\kappa$-CB when creating any new component that is disjoint from $CC_1$. By the induction assumption, the adversary can only create components that are of Type $1/2[\ell]$ for $\ell \le 2\kappa$. This proves (ii). Next we show (i). By the \emph{continuity of the type parameter}, in order for the type parameter to go beyond $2\kappa+2$ there need to be at least two components both having type parameters not less than $2\kappa +1$. 
By (ii), this is impossible, so (i) is true.
\end{proof}
\section{Lower Bounds for Several Graph Classes}
\label{sec:classes}

In this section we establish non-existence of competitive algorithms against $\kappa$-CB adversaries for various classes of graphs. We begin by establishing a strong non-competitiveness result for $\chi$-colorable graphs for $\chi \ge 3$. This generalizes part (4) of Theorem~\ref{thm:FF-bipartite} to arbitrary algorithms.

\begin{theorem}   
\label{thm:general-chi}
$\beta(1,\chi) = \infty$ for every $\chi \ge 3$.
\end{theorem}

\begin{proof}
Since $\beta(1, \chi)$ is non-decreasing in $\chi$, it suffices to prove that $\beta(1, 3) = \infty$. Fix an arbitrary coloring algorithm $\cA$. We show that for every $t \in \mathbb{N}$, a $1$-CB adversary can construct a $3$-colorable graph $G$ so that $\cA$ uses at least $t$ different bins to color vertices in $G$. It may be helpful to consult Figure~\ref{fig:lb-chi3} while reading this proof.

The construction of $G$ proceeds in layers, which we denote by $L_1, L_2, \ldots, L_{t-1}$. Vertices (and their pre-neighborhoods) in $L_1$ are presented first, followed by $L_2$, and so on. The construction stops as soon as $\cA$ uses $t$ distinct bins, which may happen before $L_{t-1}$ and will be guaranteed to happen in $L_{t-1}$.

Each layer consists of ``sufficiently many'' vertices, meaning that there should be enough vertices in lower layers to guarantee that the construction of higher layers goes through. Initially, we don't quantify ``sufficiently many,'' although we shall give some estimates on sizes of layers at the end of this proof.

Layer $L_1$ is simply a path $P$ of sufficiently large length $\ell_1$. The adversary presents the vertices in $P$ in the order in which they appear on the path. There are two possibilities: (i) $\cA$ already uses at least $t$ bins to color $P$; (ii) $\cA$ uses fewer than $t$ bins to color $P$. In case (i) the construction is over and the adversary has achieved its goal. 

Next, we handle case (ii). Observe that $\cA$ has to use at least two different bins to color $P$ correctly. Consider two bins $b_1$ and $b_2$ with the most number of vertices assigned to them by $\cA$. Let the sets of nodes assigned to those bins be $B_1$ and $B_2$, respectively, with $|B_1| \ge |B_2|$. The definition of case (ii) implies that $|B_1| \ge \ell_1/t$. Since $P$ is a path, no bin can contain more than $\ell_1/2+1$ nodes. Thus, the number of nodes not in $B_1$ is at least $\ell_1/2-1$. Since they are partitioned among at most $t-1$ bins, including $B_2$, and $B_2$ is most populous then $|B_2| \ge (\ell_1/2-1)/(t-1) = (\ell_1-2)/(2t-2)$. Next, we select subsets $B_1' \subseteq B_1$ and $B_2' \subseteq B_2$ so that all the nodes in $B_1' \cup B_2'$ are non-adjacent in $P$ and $|B_1'|=|B_2'|=\ell_1/(10t)$. This can be done as follows: alternatively pick a node from $B_1$ or $B_2$ to include in $B_1'$ or $B_2'$, respectively, and remove its neighbors from $B_2$ or $B_1$, respectively. Each pair of such steps includes one vertex into $B_1'$ and one vertex into $B_2'$ removing at most $3$ vertices from each $B_1$ and $B_2$ from future considerations. Thus, this can go on for at least $|B_2|/3 \ge \ell_1/(10t)$ rounds. In conclusion, we end up with sets of nodes $B_1'$ and $B_2'$ such that
\begin{itemize}
    \item all nodes in $B_i'$ are placed in bin $b_i$ by $\cA$, where $i \in \{1, 2\}$;
    \item $|B_1'|=|B_2'|=\ell_1/(10t)$.
\end{itemize}
In particular, the second item implies that $|B_1'|$ and $|B_2'|$ can be assumed to be sufficiently large.

Construction of each following layer $L_i$ for $i \ge 2$ either terminates early because $\cA$ used at least $t$ different bins or forces $\cA$ to assign sufficiently many vertices to bin $b_{i+1}$. We shall denote the set of such vertices\footnote{Note that the index of $B_{i+1}'$ is off by one with respect to the index of layer $L_i$ with which it is associated. This happens for $i \ge 2$ since layer $L_1$ has two sets $B_1'$ and $B_2'$ associated with it.} $B_{i+1}'$ for layer $L_i$. Assuming that the construction hasn't terminated in layer $L_{i-1}$, the next layer $L_i$ is constructed by the adversary by repeating the following steps sufficiently many times:



\begin{itemize}
    \item[(1)] the adversary chooses vertices $u_j \in B_j'$ for all $j \le i$ arbitrarily;
    \item[(2)] the adversary presents a new vertex $v$ with pre-neighborhood $\{u_1, \ldots, u_{i}\}$;
    \item[(3)] the adversary updates $B_j' \gets B_j' \setminus\{u_j\}$ for all $j \le i$.
\end{itemize}
Due to step (3) we say that $v$ \emph{consumes} nodes $u_j$ from $B_j'$ for $j \in [i]$. Observe that step (2) guarantees that $\cA$ has to assign $v$ to a bin other than $b_1, \ldots, b_i$. Just as for layer $L_1$, if $\cA$ uses $t$ different bins in this layer then we are done. Otherwise, let $b_{i+1}$ be the bin that has the most number of vertices assigned to it in layer $L_i$. If the adversary presents $\ell_i$ vertices in layer $L_i$ then the number of vertices assigned to $b_{i+1}$ is at least $\ell_i/t$. We let $B_{i+1}'$ be an arbitrary subset of such vertices of size exactly $\ell_i/t$.

This construction continues until layer $L_{t-1}$ where the adversary can present a single node according to the above scheme forcing $\cA$ to assign it to a new bin $b_t$. Overall, $\cA$ then uses $t$ different bins, namely, $b_1, \ldots, b_t$.

To guarantee that step (1) in the above construction always works, we need to make sure that all sets $B_j'$ are sufficiently large for this construction to reach layer $L_{t-1}$. This is possible provided that for $i \ge 2$ we have $|B_{i+1}'| \ge \sum_{j=i+1}^{t-1} \ell_j$ since each node in a layer above $i$ consumes one node from $B_i'$ (step (3) of the above construction). We also need a similar condition for layer $1$, namely, that $|B_2'|=|B_1'| \ge \sum_{j=2}^{t-1} \ell_j$. Thus, we end up with the following system of inequalities:
\begin{itemize}
    \item $\ell_{t-1} = 1$;
    \item $\ell_i/t \ge \sum_{j=i+1}^{t-1} \ell_j$ for $i \in \{2, 3, \ldots, t-2\}$;
    \item $\ell_1/(10t) \ge \sum_{j=2}^{t-1} \ell_j$.
\end{itemize}

It is straightforward to check that $\ell_{t-1}=1$,  
$\ell_i = t(t+1)^{t-i-2}$ for $i \in [2, t-2]$ and $\ell_1 = 10t(t+1)^{t-3}$  is a valid solution to the above system. Thus, a feasible construction can be carried out by the adversary. The total number of nodes in this construction is at most $20t(t+1)^{t-3}$.

Observe that the construction clearly satisfies the $1$-CB constraint, since layer $L_1$ is presented as a single connected component and every vertex in a higher layer is adjacent to a vertex in layer $L_1$.

We also note that the construction creates an \emph{almost-forest}. More specifically, call the vertex $v$ in step (2) of the above construction the \emph{parent} of the corresponding $u_j$ for $j \in [i]$ chosen in step (1). Observe that due to step (3), each vertex has at most one parent. Therefore, the only thing preventing this construction from being a forest is layer $L_1$, which can be thought of as a path going through all the leaves of the forest. This implies that the constructed graph is $3$-colorable. Consider the subgraph obtained by removing all edges in layer $L_1$ along with all vertices 
that do not have parents. Since it is a forest, it is $2$-colorable. Moreover, any valid $2$-coloring of this subgraph is also a partial valid $2$-coloring of the entire graph since we chose $B_1'$ and $B_2'$ to be non-adjacent. We can then extend this partial coloring to a complete $3$-coloring of the entire graph by using a greedy strategy. Note that uncolored vertices in $L_1$ have degree at most $2$, so a greedy coloring would use at most $3$ colors.
\end{proof}

\begin{figure}[h]
\centerline{\includegraphics[scale=0.6]{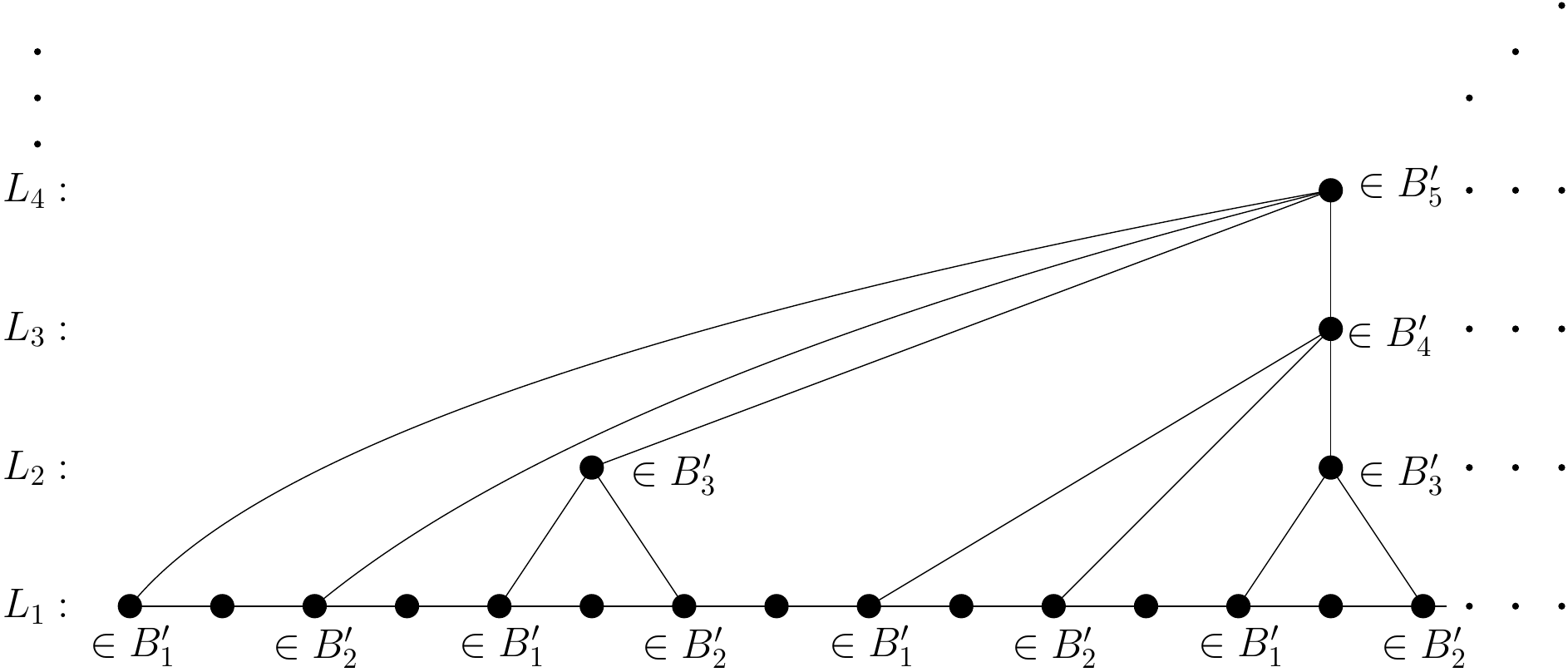}}
  \caption{Example of the construction used in Theorem~\ref{thm:general-chi}. This is a hypothetical example for some $\cA$ that assigns bins to vertices according to the figure.}
  \label{fig:lb-chi3}
\end{figure}

The above construction can be modified so that either  $\cA$ uses $t$ bins or the adversary can successively force saturated bins. For example, the adversary can extend the level $L_1$ and repeat the construction on the extended part. The adversary can do this sufficiently many times and  recolor each copy so that saturated bins are forced.

The rough estimates on sufficient lengths of layers presented in the above proof immediately lead to the following quantitative version of the result.

\begin{corollary}   
\label{cor:lowerbound}
The $1$-CB adversary can construct a $3$-colorable graph on $n$ vertices so that any online coloring algorithm uses at least $\Omega({\log n/\log\log n})$ bins.
\end{corollary}

Next, we note that the construction from Theorem~\ref{thm:general-chi} is quite robust. It can be modified in various ways to obtain similar non-competitiveness results for other classes of graphs. We first define the relevant classes. 

\begin{description}
\item[$C_k\textnormal{-}\free$:] the class of graphs that do not contain cycle of length $k$ as a (not necessarily induced) subgraph. 

\item[$d\textnormal{-}\inductive$:] the class of $d$-inductive graphs, i.e., those graphs whose vertices can be numbered so that each vertex has at most $d$ adjacent vertices among higher numbered vertices.

\item[$\planar$:] the class of planar graphs.

\item[$\treewidth\textnormal{-}k$:] the class of graphs of treewidth at most $k$.
\end{description}

We are now ready to state and prove the following corollary of the construction from Theorem~\ref{thm:general-chi}.

\begin{corollary}
\label{cor:classes_lb}
\hspace{1cm}
\begin{enumerate}
    \item $\beta(1,C_k\textnormal{-}\free) = \infty$ for every $k \ge 3$.
    \item $\beta(1,d\textnormal{-}\inductive) = \infty$ for every $d \ge 2$.
    \item $\beta(1, \planar) = \infty$.
    \item $\beta(1, \treewidth\textnormal{-}k) = \infty$ for every $k \ge 5$.
\end{enumerate}
\end{corollary}
\begin{proof}

\hspace{1cm}

\begin{enumerate}
    \item Since the construction in Theorem~\ref{thm:general-chi} is an almost-forest, the only cycles present are those using edges in layer $L_1$. By making $L_1$ longer we could insist that nodes in $B_1'$ and $B_2'$ are at least distance $k$ apart: modify the procedure for selecting nodes into $B_1'$ or $B_2'$ by picking a node from $B_1$ or $B_2$ respectively and removing all nodes at distance $k$ from it from $B_1$ and $B_2$. This modification insures that all cycles are of length greater than $k$.
    \item Observe that the construction is $2$-inductive: number vertices in the order in which they appear. Each vertex in layer $L_i$ for $i \ge 2$ has at most one neighbor among higher numbered vertices, namely, the vertex which we called the parent. A vertex in $L_1$ potentially has $2$ adjacent higher numbered vertices: at most one parent in layer $L_i$ for $i \ge 2$ and at most one neighbor in layer $L_1$ which follows it in the path.
    \item Since the construction in Theorem~\ref{thm:general-chi} is an almost-forest and the forest part is planar, we just need to make sure that the path in $L_1$ does not break planarity. We could draw a plane embedding of the forest part and order leaves clockwise. If the vertices in $L_1$ appear in the order consistent with this clockwise ordering of leaves then planarity can be maintained while adding nodes and edges from $L_1$ back into the picture. Unfortunately, the clockwise ordering of leaves in plane embedding might be inconsistent with the ordering of these leaves along the path in $L_1$. Fortunately, it is possible to adjust the construction to guarantee that the two orders are consistent.  Completely formal proof of this is rather tedious, so we give a high level description instead.
    
    First, note that there is a single tree $T$ such that for every algorithm $\cA$ the forest-part of the construction produced for $\cA$ is a subgraph of $T$, where the leaves are labelled as either $B_1'$ or $B_2'$ vertices. Second, we could consider the plane embedding of $T$ and the clockwise ordering of leaves induces a sequential pattern of inter-mixed labels $B_1'$ and $B_2'$. Third, observe that by letting the path $P$ in $L_1$ be sufficiently long and taking a subset of $B_1$ and $B_2$ appropriately, any sequential pattern of inter-mixed labels $B_1'$ and $B_2'$ can be generated along the path $P$ in $L_1$. Therefore, the adversary can always generate $B_1'$ and $B_2'$ respecting the same sequential pattern as induced by the clockwise ordering of leaves in the plane embedding of $T$. This is how the adversary generates $L_1$ in the modified construction. The adversary proceeds generating the subgraph of $T$ as before, but it uses $T$ as a guide: vertex $v$ from step (2) of the construction can be mapped to a node in $T$ and the children of $v$ in $T$ dictate which nodes $\{u_j\}$ are chosen in step (1) of the construction. An illustration is given in Figure~\ref{fig:lb-planar}. This completes the argument.
    \item Observe that the modified construction from the previous item is $2$-outerplanar since after removing the vertices in $L_1$ we are left with a graph where every vertex is adjacent to the unbounded face. Therefore by the result of Bodlaender~\cite{Bodlaender1998} the construction has treewidth at most $5$.
\end{enumerate}
\end{proof}

\begin{figure}[h]
\centerline{\includegraphics[scale=0.6]{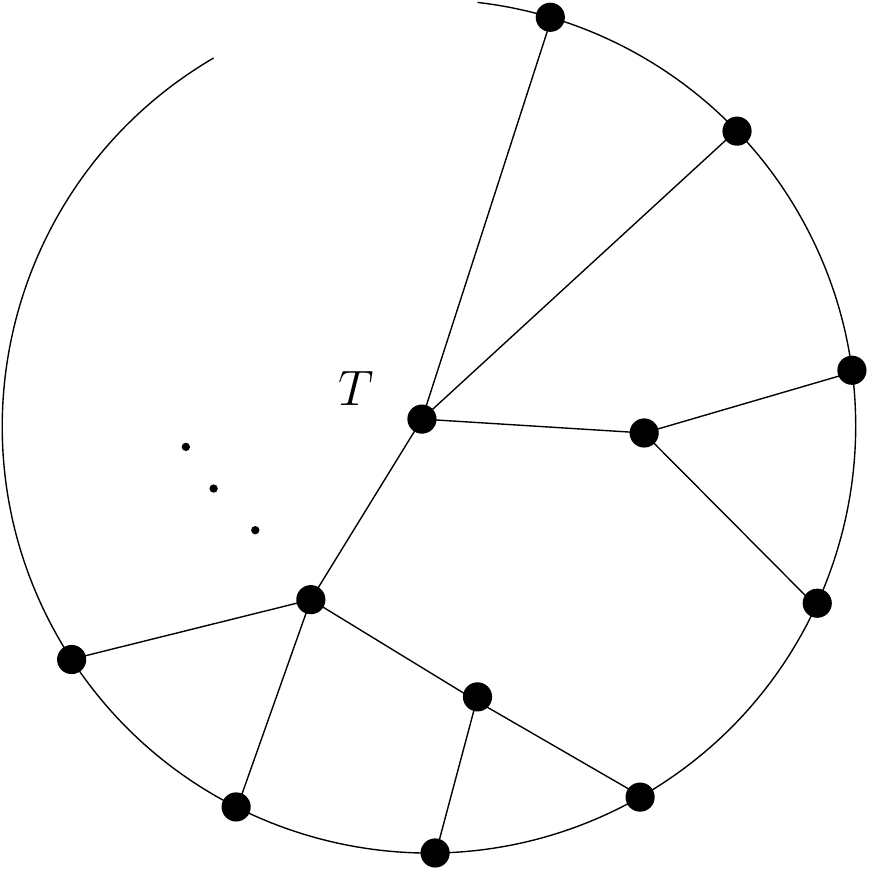}}
  \caption{Example of the construction used in parts 3 and 4 of Corollary~\ref{cor:classes_lb}. The vertices appearing on the circle are precisely the $B_1'$ and $B_2'$ subsets of $L_1$ vertices. Other $L_1$ vertices are not shown, but they can be visualized as being interspersed between them. Observe that the graph is planar and $2$-outerplanar.}
  \label{fig:lb-planar}
\end{figure}

\section{Conclusion and Open Problems}
\label{sec:conclusion}

We have introduced a new type of adversary for online graph problems and studied online coloring with respect to this adversary. This led to an improved understanding of the properties of the two widely studied online coloring algorithms $FirstFit$ and $CBIP$. Furthermore, when the adversary is $\kappa$-CB for $\kappa = O(1)$, Theorems~\ref{thm:CBIP-tree} and~\ref{thm:general-chi}  show a sharp contrast between bipartite graphs, for which the $CBIP$ uses only $O(1)$ bins, and $3$-colorable graphs for which any algorithm has to use infinitely many bins. While our work suggests many directions for future research, we find the following questions particularly intriguing:

\begin{enumerate}
\item What is $\beta(1, \treewidth\textnormal{-}k)$ for $k \in \{2, 3, 4\}$?

\item What is the performance of the algorithm in Lovasz et al.~\cite{lovasz1989line}  under the $\kappa$-CB adversary? 

\item We allow the adversary to use an unlimited number of vertices. A natural extension of our work is to study the dependence on $n$ while the adversary is $\kappa$-CB. Corollary~\ref{cor:lowerbound} is a step in that direction.   For $3$-colorable graphs, when the adversary is unconstrained,  a lower bound $\Omega(\log^2 n)$ from \cite{vishwanathan1990randomized} and an upper bound $O(n^{2/3}\log^{1/3}n)$ from \cite{kierstead1998line} are known. The upper bound from  \cite{kierstead1998line} holds for $\kappa$-CB adversary. Could Corollary \ref{cor:lowerbound} be improved to at least $\Omega(\log^2 n)$ for $3$-colorable graphs? 


\item What is the performance of randomized online coloring algorithms under the $\kappa$-CB adversary? When the adversary is unconstrained, there is a randomized algorithm \cite{vishwanathan1990randomized} that uses  $O(\sqrt{n\log n})$ bins for $3$-colorable graphs. Is it possible to improve this upper bound if the adversary is, e.g.,  $1$-CB?




\item For a graph $G$ and presentation order $\sigma$ define $\kappa(G,\sigma) = \max_i cc(G\cap \sigma([i]))$. What is the behaviour of $\kappa(G, \sigma)$ in real-world instances? As we mention in the introduction, it is expected that $\kappa(G, \sigma)$ is ``small'' for social networks. How ``small'' is it actually? What are typical values of $\kappa(G,\sigma)$? For a class of real-world instances for a particular application (such as transportation networks, social networks, or electrical networks), do $\kappa(G, \sigma)$ values follow some well-defined distribution?

\item One can study the power and limitations of the $\kappa$-CB adversary in other related models, e.g., online algorithms with advice, streaming algorithms, temporal or dynamic graphs algorithms. Interactions between various features of those models and the $\kappa$-CB constraint might lead to new algorithms or finer understanding of existing algorithms.

\item Last but definitely not least, it would be rather interesting to study other online graph problems under the $\kappa$-CB adversary.

\end{enumerate}


\bibliographystyle{plain}
\bibliography{refs}

\end{document}